\documentclass[a4paper]{article}

\usepackage{libertinus}

\usepackage{amsmath,amsthm,amssymb,latexsym}

\newtheorem{theorem}{Theorem}[section]
\newtheorem{proposition}[theorem]{Proposition}
\newtheorem{lemma}[theorem]{Lemma}

\newtheorem{remark}[theorem]{Remark}
\newtheorem{definition}[theorem]{Definition}
\newtheorem{example}[theorem]{Example}
\newtheorem{assumption}[theorem]{Assumption}

\begin{document}

\title{A note on markets with semi-static\\ trading strategies}

\author{Mikl\'os R\'asonyi\thanks{HUN-REN Alfr\'ed R\'enyi Institute of Mathematics, Budapest, 
Hungary; rasonyi@renyi.hu}}

\date{\today}

\maketitle
{}



\begin{abstract}
We consider a discrete-time financial market model where, in addition to finitely many dynamically
traded assets, there are also (possibly infinitely many) static options to choose from. We introduce 
the concept of \emph{small cones}
of random variables and present a sufficient condition for the attainable positions
in the market to be closed in probability. 

A fundamental
theorem of asset pricing is shown in the present context. Utility maximization
will also be considered. We will provide economically meaningful examples of infinite 
dimensional small cones to demonstrate the pertinence of our approach.
\end{abstract}



\section{Introduction and literature review}

In most models of financial markets, liquid assets can be bought and sold dynamically
over a trading interval. There exist, however, derivative products which are bought at 
initial date $0$ and their (random) payoff is delivered only at the final date $T>0$, with no trading between these dates.
Over-the-counter products are natural examples. 
Interestingly, such \emph{semi-static} trading
has been studied first in the context of model uncertainty. The notes \cite{hobson} are the earliest reference
that we are aware of. 

From the huge literature of robust finance, we briefly recall the setting of \cite{bice} which is the closest to our 
viewpoint. In that paper, in addition to the usual dynamic portfolios in $d$ risky assets traded at discrete time instants, 
there is a collection of functions $\phi_{i}:\mathbb{R}^{d}\to \mathbb{R}$,
$i\in I$ such that the investor may obtain, with zero cost, each of the the random payoffs $\phi_{i}(X_{T})$ where
$X_{T}$ is the terminal value of these $d$ assets and $I$ is an arbitrary index set.
So an investor's strategy consists of a sequence of portfolio rebalancings plus a choice of one of the
available static options.

In \cite{bice} a robust fundamental theorem of asset pricing was formulated for such markets in a model-free
context, under suitable technical conditions. Roughly speaking, absence of arbitrage
without reference to a specific objective probability 
was shown to be equivalent to the existence of a certain pricing measure $Q$ such that the asset price process becomes a $Q$-martingale
and also each static option $\phi_{i}(X_{T})$ has a non-positive $Q$-price. In the robust finance
literature quite a few papers have been written about similar semi-static models, we mention
\cite{hk,b,g,d,bn,nutz,nutz2}.

Surprisingly, as pointed out in Section 1 of \cite{nutz}, no similar theory has been
developed in the classical setting without model ambiguity. An explanation for this has been given by 
\cite{non-closed}: it is pivotal for 
investigations of arbitrage theory that a certain infinite-di\-men\-sio\-nal cone in a topological
vector space (the attainable positions in the given market) is closed; 
\cite{non-closed}, however, pointed out that this property \emph{fails} even in fairly benign discrete-time models.
It is important to note that in \cite{non-closed} all functionals of the terminal price $X_{T}$ were assumed to be
available as static options. We will give a detailed discussion of \cite{non-closed} in
Remark \ref{egy} below. In the multiple priors case, closedness {holds} in a suitable sense as shown
in \cite{nutz} (again, all functionals of the terminal price were considered as available options).

Our purpose in the research presented here was twofold.
On one hand, we wanted to develop a theory for classical discrete-time market models without model
ambiguity. On the other hand, in an over-the-counter market, it is unlikely that ``all imaginable
options'' are available to investors so we wanted to study arbitrage and utility maximization problems 
for various, smaller sets of static options, too.  

In this note we prove that, if the cone of available options is \emph{small} in
a technical sense to be defined, closedness of attainable positions does hold, a reasonable fundamental
theorem of asset pricing can be formulated and the corresponding utility maximization problem has a solution.
We also provide economically meaningful examples (including families of call or put options)
to which our theory applies. 

Closedness could be verified when the traded static options
form a finite-dimensional subspace, see Theorem 1.5 of \cite{bice}. We are not aware of any such result 
when static options span an infinite dimensional space (unless all the options are available as in \cite{nutz}). 
Also, there is no known fundamental theorem in the style of \cite{bice} 
for such cases and the utility maximization problem has not yet been solved.
Our paper is the first one formulating
a positive criterion ensuring closedness and its consequences.
 
In Section \ref{cones} we investigate the closedness of cones of random variables and introduce
the concept of \emph{small cones}.
In Section \ref{dynstat} we recall a market model with finitely many dynamically
traded assets and static options, but without model uncertainty. 
Absence of arbitrage in such a markets is characterized 
in Theorem \ref{nao} below, in the spirit of Theorem 1.3 of \cite{bice}, under the hypothesis that
the available set of options is a finite sum of small cones. 
Section \ref{utill} treats utility maximization in the same model.
Section \ref{exaudi} presents examples of small cones. 
Some conclusions are drawn in Remark \ref{conil}.
  
\section{Cones of random variables}\label{cones}

For $x\in\mathbb{R}$, the notation
$x^{+}$/$x^{-}$ refers to positive/negative parts.
Let $(\Omega,\mathcal{F},P)$ be the probability space we are working on, expectation is denoted by $E[\cdot]$.  
We denote by $L^{0}$ the vector space of
all (almost sure equivalence classes of) $\mathbb{R}$-valued random variables, equipped with the topology of convergence 
in probability. 
$L^{\infty}$ will stand for the set of almost surely bounded elements in $L^{0}$,
equipped with the essential supremum norm $||\cdot||_{\infty}$.
$L_{+}^{0}$ (resp. $L^{\infty}_{+}$) refers to the set of non-negative elements in $L^{0}$ (resp. $L^{\infty}$). 
For $H\subset L^{0}$, its closure
will be denoted $\overline{H}$.  For a probability $P'\ll P$ we denote by
$E_{P'}[\cdot]$ the corresponding expectation and by $L^{1}(P')$ the usual Banach-space of $P'$-integrable random variables.

A subset $C\subset L^{0}$ is called a \emph{cone} if, for all $\xi_{1},\xi_{2}\in C$ and $\lambda_{1},\lambda_{2}\geq 0$,
also $\lambda_{1}\xi_{1}+\lambda_{2}\xi_{2}\in C$. 
For a subset $H\subset L^{0}$, $\mathcal{C}(H)$ denotes
the cone generated by $H$. For $H_{1},H_{2}\subset L^{0}$
we denote $H_{1}+H_{2}:=\{h_{1}+h_{2}:h_{1}\in H_{1},h_{2}\in H_{2}\}$, $H_{1}-H_{2}$ is then self-explanatory.
A subset $H\subset L^{0}$ is called \emph{bounded} if
$$
\sup_{\xi\in H}P(|\xi|\geq n)\to 0\mbox{ as }n\to\infty.
$$
We recall two beautiful theorems that
are crucial for the present paper. 

\begin{lemma}\label{bou} If $H\subset L^{0}$ is convex and bounded 
then there is $R\sim P$
with $dR/dP\in L^{\infty}$ and $\sup_{h\in H}E_{R}[|h|]<\infty$. 
\end{lemma}
\begin{proof} In \cite{yan} this is proved under the supplementary condition $0\in H$. To see the general case,
pick some $h_{0}\in H$ and consider $\tilde{H}:=\{h-h_{0}:h\in H\}$ which is also convex, bounded and contains $0$.
By the result of \cite{yan}, there is $\tilde{R}$ with $d\tilde{R}/dP$ bounded and 
$\sup_{h\in \tilde{H}}E_{\tilde{R}}[|h|]<\infty$. Now define $dR/d\tilde{R}:=e^{-|h_0|}/E[e^{-|h_0|}]\leq M_{0}<\infty$.{}
Clearly, $dR/dP$ is bounded and
$$
\sup_{h\in H}E_{R}[|h|]\leq \sup_{h\in \tilde{H}}E_{R}[|h|+|h_{0}|]\leq M_{0}\sup_{h\in \tilde{H}}E_{\tilde{R}}[|h|]+
E_{R}[|h_{0}|]<\infty,
$$
by the choice of $\tilde{R}$ and $R$.  
\end{proof}

\begin{lemma}\label{komlos} Let $H\subset L^{0}$ be bounded and convex. Let $\xi_{n}\in H$, $n\in\mathbb{N}$.
Then there exists a subsequence $n_{k}$ and $\xi_{*}\in L^{0}$ such that
$$
\frac{\xi_{n_{1}}+\ldots+\xi_{n_{k}}}{k}\to \xi_{*}
$$
almost surely. 
\end{lemma}
\begin{proof} Apply Lemma \ref{bou}. Since
$H$ is bounded in $L^{1}(R)$, Theorem 1 of \cite{komloss} implies the statement.
\end{proof}

We introduce a rather technical concept which, however, will prove to be mathematically convenient.

\begin{definition}\label{deff} A set $L\subset L^{0}$ is called a \emph{small cone} if there exists 
a convex, closed and bounded
set $H$ with $0\notin H$ such that  $L=\mathcal{C}(H)$. Such an $H$ is called a \emph{generator} of $L$.
\end{definition}

It is a well-known fact that, in a finite-dimensional Euclidean space, a cone generated by a closed, bounded
set away from the origin is closed (see Corollary 9.6.1 of \cite{rockafellar}). In the current, 
infinite dimensional setting, small cones have similar, favourable closedness properties.

\begin{lemma}\label{comp}
Every small cone is closed. If $L$ is a small cone and $C\subset L^{0}$ is an arbitrary closed cone with $C\cap (-L)=\{0\}$
then $C+L$ is also closed.	
\end{lemma}
\begin{proof} The first statement follows from the second with the choice $C=\{0\}$ so we only prove the latter.
Let $H$ be a generator of $L$ as in Definition \ref{deff} and let $R$ be as given by Lemma \ref{bou} for this choice of $H$,
set $S:=\sup_{h\in H}E_{R}[|h|]<\infty$. 

Let $c_{n}+\alpha_{n}h_{n}\to \zeta$ in probability as $n\to\infty$, where $c_{n}\in C$, $h_{n}\in H$, $\alpha_{n}\geq 0$. If
$\sup_{n}\alpha_{n}<\infty$ then
there is a subsequence $n_{k}$ such that both $\alpha_{n_{k}}\to\alpha_{*}\geq 0$, $k\to\infty${}
and $\sum_{j=1}^{k}h_{n_{j}}/k\to h_{*}\in H$ hold almost surely, by Lemma \ref{komlos}. Then for every $\varepsilon${}
there is $N(\varepsilon)\geq 1$ such that for $j\geq N(\varepsilon)$, $|\alpha_{n_{j}}-\alpha_{*}|\leq \varepsilon$.
For $k\geq N(\varepsilon)+1$,
\begin{eqnarray*}& &
\left|\frac{1}{k}\sum_{j=1}^{k}\alpha_{n_{j}}h_{n_{j}}-\alpha_{*}h_{*}\right|\\ &\leq&{}
\left|\frac{1}{k}\sum_{j=1}^{k}\alpha_{n_{j}}h_{n_{j}}-\frac{1}{k}\sum_{j=1}^{k}\alpha_{*}h_{n_{j}}\right|+{}
\left|\frac{1}{k}\sum_{j=1}^{k}\alpha_{*}h_{n_{j}}-\alpha_{*}h_{*}\right|\\
&\leq& \left|\frac{1}{k}\sum_{j=N(\varepsilon)+1}^{k}\alpha_{n_{j}}h_{n_{j}}-\frac{1}{k}\sum_{j=N(\varepsilon)+1}^{k}\alpha_{*}h_{n_{j}}\right|
+\frac{1}{k}\left|\sum_{j=1}^{N(\varepsilon)}\alpha_{n_{j}}h_{n_{j}}\right|\\
&+& \frac{\alpha_{*}}{k}\left|\sum_{j=1}^{N(\varepsilon)}h_{n_{j}}\right| +
\alpha_{*}
\left|\frac{1}{k}\sum_{j=1}^{k}h_{n_{j}}-h_{*}\right|.
\end{eqnarray*}
For fixed $\varepsilon$, the third and fourth terms tend to $0$ in probability as $k\to\infty$. 
As to the second term,
$$
E_{R}\left[\left|\sum_{j=1}^{N(\varepsilon)}\alpha_{n_{j}}h_{n_{j}}\right|\right]/k\leq N(\varepsilon)S\sup_{n}\alpha_{n}/k,
$$
so this term also tends to $0$ in probability when $k\to\infty$ (keeping $\varepsilon$ fixed).
As to the first term,
\begin{eqnarray*}
& & E_{R}\left[\left|\sum_{j=N(\varepsilon)+1}^{k}\alpha_{n_{j}}h_{n_{j}}-\sum_{j=N(\varepsilon)+1}^{k}\alpha_{*}h_{n_{j}}\right|\right]\\
&\leq& \sum_{j=N(\varepsilon)+1}^{k} |\alpha_{n_{j}}-\alpha_{*}|E_{R}[|h_{n_{j}}|]\leq \varepsilon S k	
\end{eqnarray*}
so by Markov's inequality, for every $\eta>0$,
\begin{eqnarray*} 
P\left(\left|\frac{1}{k}\sum_{j=N(\varepsilon)+1}^{k}\alpha_{n_{j}}h_{n_{j}}-\frac{1}{k}\sum_{j=N(\varepsilon)+1}^{k}
\alpha_{*}h_{n_{j}}\right|\geq \eta{}
\right) &\leq& \varepsilon S/\eta,
\end{eqnarray*}
which tends to $0$ as $\varepsilon\to 0$.
It follows from these arguments that $\frac{1}{k}\sum_{j=1}^{k}\alpha_{n_{j}}h_{n_{j}}$ tends
to $\alpha_{*}h_{*}$ in probability.
Since
$$
\frac{1}{k}\sum_{j=1}^{k}c_{n_{j}}+\frac{1}{k}\sum_{j=1}^{k}\alpha_{n_{j}}h_{n_{j}}\to \zeta,\ k\to\infty,
$$ 
necessarily
$$
\frac{1}{k}\sum_{j=1}^{k}c_{n_{j}}\to c_{*},
$$
in probability for some $c_{*}\in C$, since $C$ is closed. But then $\zeta=c_{*}+\alpha_{*}h_{*}\in C+L$.

In the case $\sup_{n}\alpha_{n}=\infty$ there is a subsequence $n_{k}$ with $0<\alpha_{n_{k}}\to\infty$, $k\to\infty$ and
$\sum_{j=1}^{k}h_{n_{j}}/k\to h_{*}$ almost surely for some $h_{*}\in H$. Since
$$
\frac{c_{n_{k}}}{\alpha_{n_{k}}}+h_{n_{k}}\to 0,\ k\to\infty,
$$
also
$$
\frac{1}{k}\sum_{j=1}^{k}\frac{c_{n_{j}}}{\alpha_{n_{j}}}+\frac{1}{k}\sum_{j=1}^{k}h_{n_{j}}\to 0,\ k\to\infty,
$$
thus $\frac{1}{k}\sum_{j=1}^{k}\frac{c_{n_{j}}}{\alpha_{n_{j}}}\to -h_{*}\in C$ in probability, as $k\to\infty$.{}
But then $$-h_{*}\in C\cap (-H)\subset C\cap (-L)=\{0\},$$ which contradicts $0\notin H$.
We conclude that $\sup_{n}\alpha_{n}=\infty$ is not possible hence $C+L$ is closed.
\end{proof}


\begin{remark}\label{notall}{\rm Clearly, not all cones are small.  
One may consider $L^{\infty}_{+}\subset L^{0}$. This is generated
by the convex set $H:=\{g\in L^{0}_{+}:||g||_{\infty}\leq 1\}$ which is  
closed and bounded in $L^{0}$ (but contains
$0$). However, $L^{\infty}_{+}$ is clearly not closed in $L^{0}$ so it is not a small cone.

A subspace $L\subset L^{0}$ is never a small cone. Indeed, arguing by contradiction, let $H$ be a generator of $L$. Take 
an arbitrary $h_{0}\in H$. If $L$ is a subspace then also $-h_{0}\in L$ hence $h_{0}=-\alpha h_{1}$ for some $h_{1}\in H$
and $\alpha\geq 0$. Then, by convexity of $H$, $\frac{h_{0}+ \alpha h_{1}}{1+\alpha}=0\in H$,
which contradicts Definition \ref{deff}.} 
\end{remark}

\section{Market with dynamically traded assets and static options}\label{dynstat}

Let $\mathcal{F}_{t}$, $t=0,\ldots, T$ be a discrete-time filtration on the given probability space, $\mathcal{F}_{0}$
is the trivial sigma-algebra. Let $X_{t}$, $t=0,\ldots,T$ be an adapted
$\mathbb{R}^{d}$-valued process describing the discounted prices of $d$ dynamically traded assets.
We denote by $\langle\cdot,\cdot\rangle$ the standard scalar product in $\mathbb{R}^{d}$.

A trading strategy is a $\mathbb{R}^{d}$-valued process $\phi_{t}$, $t=1,\ldots,T$ such that
$\phi_{t}$ is $\mathcal{F}_{t-1}$-measurable, the set of all such strategies is denoted $\Phi$. 
We further assume that a set of zero-price static options $\mathcal{Y}\subset L^{0}$
is given. These are not liquid, an investor with initial capital $0$ can choose one element 
$h\in\mathcal{Y}$ at time $0$ and then his/her
portfolio value at time $t=0,\ldots,T$ will be
$$
V(\phi,h):=h+\sum_{j=1}^{t}\langle \phi_{j},X_{j}-X_{j-1}\rangle.{}
$$
We will also write $W_{t}(\phi):=\sum_{j=1}^{t}\langle \phi_{j},X_{j}-X_{j-1}\rangle$, $0\leq t\leq T$.

Denote $$
K_{0}:=\{V(\phi,0):\phi\in \Phi\}=\{W_{T}(\phi):\phi\in\Phi\},$$ 
the set of portfolio values in the dynamically traded assets.
Then $K_{0}+\mathcal{Y}$ describes all attainable positions from $0$ initial capital.

We say that $\mathbf{NA}$ holds is $K_{0}\cap L_{+}^{0}=\{0\}$. We say that $\mathbf{NA}(\mathcal{Y})$ holds
if $$
\left(K_{0}+\mathcal{Y}\right)\cap L_{+}^{0}=\{0\}.$$

$\mathcal{M}$ is the set of probablities $Q\sim P$ on $(\Omega,\mathcal{F})$ such that $X_{t}$, $t=0,\ldots,T$ is
a $Q$-martingale with respect to the given filtration. Define $$
\mathcal{M}_{\infty}=\{Q\in \mathcal{M}:dQ/dP\in L^{\infty}\}.$$
Analogously, $$
\mathcal{M}(\mathcal{Y}):=\{Q\in\mathcal{M}:\mbox{ for all }Y\in\mathcal{Y},Y\in L^{1}(Q)\mbox{ and }E_{Q}[Y]\leq 0\}
$$
and $\mathcal{M}_{\infty}(\mathcal{Y}):=\{Q\in\mathcal{M}(\mathcal{Y}):dQ/dP\in L^{\infty}\}$.

$\mathbf{NA}$ implies that the set $K_{0}-L_{+}^{0}$
is closed in $L^{0}$, see Theorem 6.9.2 of \cite{ds}; $\mathbf{NA}$ 
is equivalent to $\mathcal{M}_{\infty}\neq\emptyset$,
see Theorem 6.1.1 of \cite{ds}. 

It is clear that $\mathcal{M}(\mathcal{Y})\neq \emptyset\Rightarrow \mathbf{NA}(\mathcal{Y})$.
One may cherish the hope that $\mathbf{NA}(\mathcal{Y})$ is equivalent to
$\mathcal{M}(\mathcal{Y})\neq \emptyset$. The situation, however, is not so simple.

\begin{example}\label{nohope}{\rm Let $d=T=1$, $X_{0}=0$, $X_{1}$ a standard Cauchy random variable,
$\mathcal{F}:=\mathcal{F}_{1}=\sigma(X_{1})$; let $\mathcal{Y}$ be the set of zero-mean random
variables. If $Q\sim P$ and $E_{Q}[Y]\leq 0$ for all $Y\in\mathcal{Y}$ then necessarily $Q=P$. But
$P\notin\mathcal{M}(\mathcal{Y})$ since $X_{1}$ is not $P$-integrable. Hence $\mathcal{M}(\mathcal{Y})=\emptyset$.{}
We show that, nevertheless, $\mathbf{NA}(\mathcal{Y})$ holds. 

Suppose we had $wX_{1}+Y\geq 0$
for some $Y\in \mathcal{Y}$ and $w>0$.  
Then $X_{1}\geq -Y/w$ so $X_{1}^{-}$ would be integrable which is nonsense.
$w<0$ leads to a similar contradiction. Hence $w=0$ but then $E_{P}[Y]=0$
implies $Y=0$ a.s.}  
\end{example}

Example \ref{nohope} shows that we need additional hypotheses on $\mathcal{Y}$ for the equivalence
\begin{equation}\label{ezjo}
\mathcal{M}(\mathcal{Y})\neq\emptyset\iff \mathbf{NA}(\mathcal{Y})	
\end{equation} to hold. Note that
a proof would probably require the use of the Kreps-Yan 
separation theorem which needs, roughly speaking, that 
the cone ${K}_{0}+\mathcal{Y}-L_{+}^{0}$ is closed. 


\begin{remark}\label{egy}{\rm We briefly review certain results of \cite{non-closed} in order to put our achievements in context.
In that paper $P\in\mathcal{M}$ was assumed and, instead of $K_{0}$, the set 
$$
\tilde{K}_{0}:=\{W_{T}(\phi):(W_{t}(\phi))_{0\leq t\leq T}\mbox{ is a supermartingale}\}\subset K_{0}
$$
was considered. The authors took $\mathcal{Y}=L^{1}(\Omega,\sigma(X_{T}),P)$ (the set of integrable 
$\sigma(X_{T})$-measurable random variables). They exhibited (in the case $d=1$, $T=2$) a bounded price
process for which $\tilde{K}_{0}-L_{+}^{0}+\mathcal{Y}$ failed to be closed in $L^{0}$ 
(this is not surprising, $\mathcal{Y}$ not being closed in $L^{0}$ either), even 
$$
(\tilde{K}_{0}+\mathcal{Y}-L_{+}^{0})\cap L^{1}(P)=[(\tilde{K}_{0}-L_{+}^{0})\cap L^{1}(P)]+\mathcal{Y}
$$ 
failed to be closed in $L^{1}(P)$. The sequences of random variables
they constructed had further specific properties that we do not list here.

The proof of \cite{non-closed} actually shows that $({K}_{0}+\mathcal{Y}-L_{+}^{0})\cap L^{1}(P)$ is not
closed either. 
Indeed, in the last step of the proof in Section 3 of \cite{non-closed},
one takes $0\leq g\leq u+v$ with $u\in \tilde{K}_{0}$, $v\in\mathcal{Y}$. We point out that
even if we take $u$ from the larger set $K_{0}$, necessarily $-v\leq g-v\leq u$ and, since $v$ is integrable, 
it \emph{follows} that $u\in \tilde{K}_{0}$ (a martingale transform with terminal value
bounded from below by an integrable random variable is, in fact, a martingale, see Theorems 1 and 2 of \cite{jacod}) 
hence the proof of \cite{non-closed} goes through
with $\tilde{K}_{0}$ replaced by $K_{0}$.

There is, however, a considerable conceptual difference between \cite{non-closed} and the present paper.
We regard $\mathcal{Y}$ as a family of options available at $0$ cost while the choice
in \cite{non-closed} is not (and cannot be) interpreted in this way.} 
\end{remark}

Our purpose in the present paper is to
provide conditions on $\mathcal{Y}$ that imply closedness of $K_{0}+\mathcal{Y}-L_{+}^{0}$ 
in meaningful cases. The sets $\mathcal{Y}$ in Example \ref{nohope} and in Remark \ref{egy} were rather large so choosing
smaller sets might lead to success. We propose the following
hypothesis.

\begin{assumption}\label{asi} $\mathcal{Y}=\sum_{j=1}^{m}C_{j}$ where
each $C_{j}$ is a small cone and $$
\left(K_{0}+\sum_{j=1}^{n}C_{j}\right)\cap \left(-C_{n+1}\right)=\{0\}
$$
for all $0\leq n\leq m-1$.
\end{assumption}

\begin{remark}\label{ellen} {\rm Recall that $K_{0}$ is a vector space so in the case $m=1$, $K_{0}\cap (-C_{1})=\{0\}$ 
is equivalent to $K_{0}\cap C_{1}=\{0\}$ which means that the non-zero static options should not
be replicable from $0$ initial capital. Hence this is a sort of no-redundancy hypothesis.}
\end{remark}

Our main result comes next: a fundamental theorem of asset pricing in the present setting.

\begin{theorem}\label{nao} Let $\mathcal{Y}\subset L^{0}$ be arbitrary. Consider the following statements.
\begin{enumerate}
	
\item $\mathbf{NA}(\mathcal{Y})$ holds.

\item For each $Y\in\mathcal{Y}$ there is $Q(Y)\in \mathcal{M}_{\infty}$
such that $Y$ is $Q(Y)$-integrable and $$
E_{Q(Y)}[Y]\leq 0.
$$

\item $\mathcal{M}_{\infty}(\mathcal{Y})\neq \emptyset$.

\item The set $K_{0}+\mathcal{Y}-L^{0}_{+}$ is closed.

\end{enumerate}  
Then $3.\to 2.\leftrightarrow 1.$ always. 
If Assumption \ref{asi} holds then $1.\to 4.$ and also $3.\leftrightarrow 2.\leftrightarrow 1.$
\end{theorem}
\begin{proof} Suppose that $2.$ holds and $V(\phi,Y)\geq 0$ almost surely for some
$(\phi,Y)\in \mathcal{A}$. Notice that $W_{T}(\phi)\geq -h$ in this case and, $Y$
being $Q(Y)$-integrable and $(X_{t})_{0\leq t\leq T}$ being a $Q(Y)$-martingale,
the $Q(Y)$-martingale transform $(W_{t}(\phi))_{0\leq t\leq T}$ has a final
value that is above a $Q(Y)$-integrable random variable hence 
$(W_{t}(\phi))_{0\leq t\leq T}$ is a true $Q(Y)$-martingale by Theorems 1 and 2 of \cite{jacod}.
Then $E_{Q(Y)}[W_{T}(\phi)]=0$ and $E_{Q(Y)}[Y]\leq 0$. Consequently, $E_{Q(Y)}[V(\phi,Y)]\leq 0$ and
$V(\phi,Y)=0$ almost surely: $\mathbf{NA}(\mathcal{Y})$ holds.

Suppose that $1.$ is true. Let $Y\in\mathcal{Y}$ be arbitrary. Define $dP'/dP:=e^{-|Y|}/E[e^{-|Y|}]$,
clearly $Y\in L^{1}(P')$. $\mathbf{NA}(\mathcal{Y})$ implies $\mathbf{NA}$ but, equivalently,
also $\mathbf{NA}$ under the probability $P'$. Hence Theorem 6.1.1 implies the existence of
$Q_{0}\in\mathcal{M}$ with $dQ_{0}/dP'\in L^{\infty}$. If $Y\in K_{0}$ then $E_{Q_{0}}[Y]=0$
(since $Y=W_{T}(\phi)$ for some $\phi$ and it is integrable so $W_{t}(\phi)$, $t=0,\ldots,T${}
is a $Q_{0}$-martingale). $dQ_{0}/dP\in L^{\infty}$ by construction.

In the case $Y\notin K_{0}$, $L:=\{\alpha Y:\alpha\geq 0\}$ is a small cone with $K_{0}\cap (-L)=\{0\}$.
We claim that also $(K_{0}-L_{+}^{0})\cap (-L)=\{0\}$.
Indeed, if $k_{0}-l=-\alpha Y$ for some $k_{0}\in K_{0}$, $l\in L$, $\alpha\geq 0$ then
$k_{0}+\alpha Y=l\in L_{+}^{0}$ so $l=0$ by $\mathbf{NA}(\mathcal{Y})$ and $k_{0}=-\alpha Y$ which necessitates $\alpha=0$.
$K_{0}-L_{+}^{0}$ is closed by Theorem 6.9.2 of \cite{ds}.
Now by Lemma \ref{comp} (or by direct argument), $K_{0}-L_{+}^{0}+L$ is a closed cone in $L^{0}$.

Theorem 5.2.3 of \cite{ds} (with the choice $p=1$, $C:=(K_{0}-L_{+}^{0}+L)\cap L^{1}(Q_{0})$
and dual pair $E=L^{1}(Q_{0})$, $E'=L^{\infty}$) provides a probability $Q(Y)\sim P$
with $dQ(Y)/dQ_{0}\in L^{\infty}$ such that $Q(Y)\in \mathcal{M}_{\infty}$ and also
$E_{Q(Y)}[Y]\leq 0$: $2.$ follows.


$3.$ trivially implies $2.$ Now let Assumption \ref{asi} be in force and assume 1. 
which implies $\mathbf{NA}$ so $K_{0}-L_{+}^{0}$ is closed by Theorem 6.9.2 of \cite{ds}.

Apply Lemma \ref{comp} inductively. Assume that $K_{0}-L_{+}^{0}+\sum_{j=1}^{n}C_{j}$
is closed for some $0\leq n\leq m-1$. Take $\xi=k-l+\sum_{j=1}^{n}c_{j}\in K_{0}-L_{+}^{0}+\sum_{j=1}^{n}C_{j}$.
If we had $\xi\in -C_{n+1}$ then $l=k+\sum_{j=1}^{n}c_{j}-\xi\in (K_{0}+\sum_{j=1}^{n+1}C_{j})\cap L_{+}^{0}$
so $\mathbf{NA}(\mathcal{Y})$ implies $l=0$ and then $\xi\in (K_{0}+\sum_{j=1}^{n}C_{j})\cap (-C_{n+1})=\{0\}$,
by assumption. As $C_{n+1}$ is a small cone, Lemma \ref{comp} implies (with the choice
$C=K_{0}-L_{+}^{0}+\sum_{j=1}^{n}C_{j}$, $L=C_{n+1}$) that $K_{0}-L_{+}^{0}+\sum_{j=1}^{n+1}C_{j}$ is also
a closed cone. Finally, $K_{0}-L_{+}^{0}+\sum_{j=1}^{m}C_{j}$ is shown to be closed and 4. follows.

Now we prove that, under Assumption \ref{asi}, 1. implies 3.
We have already seen that $\mathbf{NA}(\mathcal{Y})$ implies $4.$ 
Let $H_{j}$ generate $C_{j}$ as in Definition \ref{deff}. 
Lemma \ref{bou} provides $R_{j}\sim P$, $j=1,\ldots,m$ with $dR_{j}/dR_{j-1}$, $j=2,\ldots, m$ 
bounded and also $dR_{1}/dP$ bounded, such that $\sup_{h\in H_{j}}E_{R_{j}}[|h|]<\infty$.{}
It follows that $E_{R_{m}}[|\xi|]<\infty$ for all $\xi\in \sum_{j=1}^{m}C_{j}$. Now
define $$
dR/dR_{m}:=\exp\left(\sum_{t=0}^{T}|X_{j}|\right)/E_{R_{m}}\left[\exp\left(\sum_{t=0}^{T}|X_{j}|\right)\right].
$$
Notice that $\langle \zeta,X_{t}-X_{t-1}\rangle$ are $R$-integrable for each $1\leq t\leq T$ 
and for each $\mathcal{F}_{t-1}$-measurable 
$d$-dimensional bounded random variable $\zeta$. 

Theorem 5.2.3 of \cite{ds} (with the choice $p=1$, $C:=(K_{0}-L_{+}^{0}+\mathcal{Y})\cap L^{1}(R)$
and dual pair $E=L^{1}(R)$, $E'=L^{\infty}$) provides a probability $Q_{*}\sim P$
such that $E_{Q_{*}}[c]\leq 0$ for all $c\in (K_{0}-L_{+}^{0}+\mathcal{Y})\cap L^{1}(R)$ 
and $dQ_{*}/dR$ is bounded (hence $dQ_{*}/dP$
is also bounded). By the discussion above, $(K_{0}-L_{+}^{0}+\mathcal{Y})\cap L^{1}(R)$ contains
all random variables of the form $\langle \zeta,X_{t}-X_{t-1}\rangle$ which implies $Q_{*}\in \mathcal{M}_{\infty}$.
The construction of $Q_{*}$ ensures that $E_{Q_{*}}[Y]\leq 0$ for all $Y\in\mathcal{Y}$, $3.$ follows.
\end{proof}

\begin{remark}{\rm As already discussed in \eqref{ezjo} above, the equivalence $1.\leftrightarrow 3.$ should be considered as the 
``right'' fundamental theorem of
asset pricing for a market with both liquid assets and static options. It asserts that there is a
``uniform'' pricing measure $Q_{*}\in\mathcal{M}_{\infty}(\mathcal{Y})$ which assigns a non-positive price
to \emph{every} the static option at the same time.
Without Assumption \ref{asi}, we can guarantee only that for each $Y\in\mathcal{Y}$ there is a pricing measure
$Q(Y)\in\mathcal{M}_{\infty}$ with $E_{Q(Y)}[Y]\leq 0$ which depends on $Y$: this is the content of $2.$ in Theorem \ref{nao}.}
\end{remark}

\section{Utility maximization with semi-static strategies}\label{utill}

In the discussion following Theorem 1.3 of \cite{non-closed} the authors point out that the failure
of the closedness of final outcomes in a market with semistatic trading strategies makes utility
maximization a delicate issue in these models. 

In our setting, it is possible to
prove the existence of optimal strategies for a utility maximizer,
an important consequence of 4. in Theorem \ref{nao}. 
Define $\mathcal{A}:=\Phi\times\mathcal{Y}$, the set of all semi-static strategies.

\begin{theorem}\label{util} Let Assumption \ref{asi} and
$\mathbf{NA}(\mathcal{Y})$ hold. Let $u:\mathbb{R}\to\mathbb{R}$ be non-decreasing, concave and bounded from above.
Then there exist $(\phi^{*},h^{*})\in\mathcal{A}$ such that
$$
\sup_{(\phi,h)\in\mathcal{A}}E[u(V(\phi,h))]=E[u(V(\phi^{*},h^{*}))],
$$
that is, the utility maximization problem with dynamic assets and static options is well-posed.	
\end{theorem}
\begin{proof} We follow ideas from Proposition 4.3 of \cite{largo}.
Let $(\phi_{n},h_{n})$, $n\in\mathbb{N}$ be a sequence such that $$
\sup_{(\phi,h)\in\mathcal{A}}E[u(V(\phi,h))]=\lim_{n\to\infty}
E[u(V(\phi_{n},h_{n}))]
$$ and $E[u(V(\phi_{n},h_{n}))]>-\infty$ for all $n$. $u$ being bounded from above, 
necessarily 
\begin{equation}\label{suppe}
\sup_n E[u^-(V(\phi_{n},h_{n}))]<\infty.	
\end{equation}
{}
The case of constant $u$ is trivial, otherwise 
there exist $c_{0},c_{1}>0$ such that $u(x)\leq -c_{0}|x|+c_{1}$ for
all $x\leq 0$, see Lemma 4.1 of \cite{largo} for this simple fact. Now
\eqref{suppe} implies
$$
\sup_n E[V^-(\phi_{n},h_{n})]<\infty.
$$ 
Proposition 2 of \cite{oldlargo} implies that there is $\tilde{P}\sim P$ with bounded $d\tilde{P}/dP$ such that
$V(\phi_{n},h_{n})$ are $\tilde{P}$-integrable for all $n$. 

Take $Q\in\mathcal{M}(\mathcal{Y})$ with $dQ/d\tilde{P}$ bounded, as constructed
in Theorem \ref{nao} (note that $\mathbf{NA}(\mathcal{Y})$ is true under $P$ iff it is true under $\tilde{P}$).
Hence $V(\phi_{n},h_{n})$ are $Q$-integrable and 
$\sup_n E_Q[V^{-}(\phi_{n},h_{n})]<\infty$. As $E_{Q}[Z]\leq 0$ for each 
$Z\in (K_{0}+\mathcal{Y}-L_{+}^{0})\cap L^{1}(\tilde{P})$ and $V(\phi_{n},h_{n})\in L^{1}(\tilde{P})$
for all $n$, it follows that also 
$$
\sup_n E_Q[V^{+}(\phi_{n},h_{n})]\leq \sup_n E_Q[V^{-}(\phi_{n},h_{n})]<\infty$$ 
so, finally, $\sup_n E_Q[|V(\phi_{n},h_{n})|]<\infty$.

Applying Theorem 1 of \cite{komloss} in $L^1(Q)$, for some subsequence $n_{k}$, $k\in\mathbb{N}$,
the C\'esaro means of $V(\phi_{n_{k}},h_{n_{k}})$ converge to some
random variable $X$ almost surely. By 4. in Theorem \ref{nao},
$X\in K_{0}+\mathcal{Y}-L_{+}^{0}$ so $X\leq V(\phi^{*},h^{*})$ for some
$(\phi^{*},h^{*})\in\mathcal{A}$. By concavity of $u$
and the limsup Fatou lemma, 
\begin{eqnarray*}
E[u(V(\phi^{*},h^{*}))]&\geq& E[u(X)]\geq \lim_{k\to\infty}E\left[u\left(V\left(\frac{1}{k}\sum_{l=1}^{k}\phi_{n_{l}},
\frac{1}{k}\sum_{l=1}^{k}h_{n_{l}}\right)\right)\right]\\
&\geq& \lim_{k\to\infty}\frac{1}{k}\sum_{l=1}^{k}E[u(V(\phi_{n_{l}},h_{n_{l}}))]=\lim_{k\to\infty}E\left[u\left(V\left(\phi_{n_{k}},
h_{n_{k}}\right)\right)\right]\\
&=&\sup_{(\phi,h)\in\mathcal{A}}E[u(V(\phi,h))],
\end{eqnarray*}
and the result follows since $$E[u(V(\phi^{*},h^{*}))]\leq \sup_{(\phi,h)\in\mathcal{A}}E[u(V(\phi,h))],$$ trivially.
\end{proof}

\begin{remark}{\rm In \cite{nutz2} the existence of an optimizer and
a duality formula were shown for the utility maximization problem in the multiprior case but only for exponential utility 
and for $T=2$.
As usual in that stream of literature, all $0$-price options were assumed to be available for investors.

Theorem \ref{util} combined with the results of Section \ref{exaudi} below implies that a utility maximizer
who trades dynamically and has access to a family of call (or put) options can find an optimal
semi-static portfolio. No similar result has been shown in the literature yet.}
\end{remark}

\section{Examples}\label{exaudi}

In this section we present several examples of small cones. For $H\subset L^{0}$, $\mathrm{conv}(H)$ denotes 
the set of finite convex combinations of elements in $H$.

\begin{example}{\rm Let $\xi_{n}$, $n\in\mathbb{N}$ be a sequence of zero-mean independent random variables
such that $J:=\sup_{n}E[\xi_{n}^{2}]<\infty$. 
Let $a_{n}$, $n\in\mathbb{N}$ be a sequence of constants satisfying $a_{\star}\leq a_{n}\leq a_{\dagger}$, $n\in\mathbb{N}$
for some $0<a_{\star}\leq a_{\dagger}$. Let $H:=\{\xi_{n}+a_{n},n\in\mathbb{N}\}$. 

For each $h\in\mathrm{conv}(H)$,
$h=\sum_{j=1}^{k}\lambda_{j}(\xi_{j}+a_{j})$ for some $k$ and convex weights $\lambda_{1},\ldots,\lambda_{k}$
so $\mathrm{var}(h)\leq J\sum_{j=1}^{k}\lambda_{j}^{2}\leq J$ and $E[|h|^{2}]\leq J+a_{\dagger}^{2}$. It follows that
$\mathrm{conv}(H)$ is uniformly integrable.

Clearly, $E[h]\geq a_{\star}$ for all $h\in \mathrm{conv}(H)$.  
If $h_{n}\in H$ converge almost surely to some $\hat{h}$ also
$E[\hat{h}]\geq a_{\star}$ by uniform integrability. This implies $0\notin H':=\overline{\mathrm{conv}(H)}$.
Also note that, for all $h\in\mathrm{conv}(H)$,
$$
P(|h|\geq l)\leq E[h^{2}]/l^{2}\leq (J+a_{\dagger}^{2})/l^{2}\to 0$$ 
as $l\to\infty$, so $\mathrm{conv}(H)$ and hence $H':=\overline{\mathrm{conv}(H)}$ are bounded in $L^{0}$. 
We conclude that $\mathcal{C}({H}')$ is a small cone.

Such an example is not artificial at all: the set of payoffs in $\mathcal{C}({H}')$
are similar to those available in the Arbitrage Pricing Theory model of \cite{ross} which 
can be used to describe one-step \emph{large financial markets}, see also \cite{oldlargo,largo}.}  
\end{example}

\begin{example}{\rm In a similar vein, if $H=\{\xi_{n}, n\in\mathbb{N}\}$ are independent with standard Cauchy
distribution then all elements of $H':=\overline{\mathrm{conv}(H)}$ have this same distribution so
${H}'$ is a closed, bounded convex set not containing $0$.  $\mathcal{C}(H')$ is thus a small cone.}
\end{example}

\begin{example}{\rm Let $\xi_{i}\in L^{0}$, $i=1,\ldots,N$ be linearly independent.
Then $$\mathcal{C}(\{\xi_{i},i=1,\ldots,N\})$$ is a small cone. Indeed, $\mathrm{conv}(\{\xi_{i},i=1,\ldots,N\})$ is closed in this case
(see e.g.\ Corollary 9.8.2 of \cite{rockafellar}) and it does not contain $0$ by linear independence.}
\end{example}

In the rest of this section, we enter the setting of Section \ref{dynstat}. 
For simplicity, we assume that the number of dynamically traded assets is $d=1$ and we will only look
at the case $m=1$ in Assumption \ref{asi}. In the two examples below
we provide infinite dimensional small cones with clear financial interpretations.

\begin{example}\label{coption}{\rm 
Let $X_{T}$ be $P$-integrable, let $0<K_{\flat}<K_{\sharp}$ and consider the call options 
$$\{(X_{T}-K)^{+}: K_{\flat}\leq K\leq K_{\sharp}\}.$$
Take a mapping $Q:[K_{\flat},K_{\sharp}]\to \mathcal{M}$ such that 
$$
D:=\sup_{K}E_{Q(K)}[(X_{T}-K)^{+}]<\infty,
$$
that is, the prices of the available call options should be bounded from above: a rather natural assumption.
Define $$H_{1}:=\{(X_{T}-K)^{+}-E_{Q(K)}[(X_{T}-K)^{+}]: K_{\flat}\leq K\leq K_{\sharp}\},$$
set $\tilde{H}_{1}:=\overline{\mathrm{conv}(H_{1})}$ and $L_{1}:=\mathcal{C}(\tilde{H}_{1})$.}
\end{example} 

\begin{proposition}\label{calloption} Assume that $X_{T}$ is unbounded from above but $P$-integrable.{}
Then the set $\tilde{H}_{1}$ is bounded and does not contain $0$ hence $L_{1}$ is a small cone.
\end{proposition}
\begin{proof}
Since $-D\leq h\leq X_{T}^{+}$ for all $h\in H_{1}$ this is also true for 
all $h\in\mathrm{conv}(H_{1})$ so $\tilde{H}_{1}$ is bounded. 

It is easy to construct an auxiliary probability ${P}^{\natural}\ll P$ 
with bounded $d{P}^{\natural}/dP$ such that 
$$E_{{P}^{\natural}}[(X_{T}-K_{\sharp})^{+}]\geq D+1.$$
Indeed, $p:=P(X_{T}-K_{\sharp}>D+1)>0$.{}
Define $dP^{\natural}/dP:=1_{\{X_{T}>D+1\}}/p$.

It follows that, for each $K\in [K_{\flat},
K_{\sharp}]$, 
$$
E_{{P}^{\natural}}[(X_{T}-K)^{+}]\geq E_{{P}^{\natural}}[(X_{T}-K_{\sharp})^{+}]\geq 
D+1\geq E_{Q(K)}[(X_{T}-K)^{+}]+1,
$$ 
hence $E_{{P}^{\natural}}[h]\geq 1$ for
all $h\in H$, even for all $h\in\mathrm{conv}(H)$. Now taking $h_{n}\in\mathrm{conv}(H)$, $h_{n}\to h_{*}$
in $L^{0}$, the limsup Fatou lemma implies $E_{{P}^{\natural}}[h_{*}]\geq 1$ since $E_{P}[X_{T}^{+}]<\infty$,{}
$d{P}^{\natural}/dP$ is bounded and hence $X_{T}^{+}$ is a ${P}^{\natural}$-integrable majorant of the sequence $h_{n}$.
\end{proof}

\begin{proposition}\label{callin} 
Assume that the $\mathcal{F}_{T-1}$-conditional support of $X_{T}$ almost surely contains $[0,\infty)$.{}
Then $K_{0}\cap L_{1}=\{0\}$.
\end{proposition}
\begin{proof} We argue by contradiction. If $K_{0}\cap L_{1}\neq \{0\}$ then there is $h_{*}\in K_{0}\cap\tilde{H}_{1}$.
Let $\lambda^{n}_{1},\ldots,\lambda^{n}_{k(n)}$ be convex weights and assume that
$$
h_{n}:=\sum_{j=1}^{k(n)}\lambda^{n}_{j}[(X_{T}-K^{n}_{j})^{+}-E_{Q(K_{j}^{n})}[(X_{T}-K^{n}_{j})^{+}]]\to h_{*}$$ 
almost surely as $n\to\infty$,{}
where $K^{n}_{j}\in [K_{\flat},K_{\sharp}]$, $j=1,\ldots,k(n)$. If one had $h_{*}\in K_{0}$ then
also $X_{T}-X_{0}-h_{*}\in K_{0}$. 
Notice that $$
K_{\sharp}+D-X_{0}\geq X_{T}-h_{n}-X_{0}\geq -X_{0},\ n\in\mathbb{N}
$$ 
so $X_{T}-X_{0}-h_{*}$ 
is a bounded random variable.

Since $X_{T}-X_{0}-h^{*}$ is in $K_{0}$, $X_{T}-X_{0}-h_{*}=\xi_{T-1}+\phi_{T-1}(X_{T}-X_{T-1})$ 
holds
for some $\mathcal{F}_{T-1}$-measurable $\xi_{T-1},\phi_{T-1}$. By our hypothesis on the conditional
support, necessarily $\phi_{T-1}=0$. In other words, $X_{T}-X_{0}-h_{*}$ and hence also 
$X_{T}-h_{*}$ are $\mathcal{F}_{T-1}$-measurable.

We will show that this is not possible. On the event $\{X_{T}\geq K_{\sharp}\}$,
$$
X_{T}-h_{n}=c_{n}:=\sum_{j=1}^{k(n)}\lambda^{n}_{j}[K^{n}_{j}+E_{Q(K_{j}^{n})}[(X_{T}-K^{n}_{j})^{+}]]
$$
so $c_{n}\to c_{*}$, $n\to\infty$ for some constant $c_{*}$.
It follows that $$
P(X_{T}-h_{*}=c_{*}|\mathcal{F}_{T-1})\geq P(X_{T}\geq K_{\sharp}|\mathcal{F}_{T-1})>0
$$
almost surely by our hypothesis on the conditional support of $X_{T}$. On the other hand,
on the event $\{0\leq X_{T}\leq K_{\flat}/2\}$, $$
h_{n}=-\sum_{j=1}^{k(n)}\lambda^{n}_{j}E_{Q(K_{j}^{n})}[(X_{T}-K^{n}_{j})^{+}]\to d_{*},\ n\to\infty{}
$$
for some constant $d_{*}$. Hence on this event $X_{T}-h_{*}\leq K_{\flat}/2-d_{*}$ holds. 
Clearly, $c_{*}\geq K_{\flat}-d_{*}$ hence $X_{T}-h_{*}<c_{*}$ on  
$\{0\leq X_{T}\leq K_{\flat}/2\}$. Notice that
$$P(0\leq X_{T}\leq K_{\flat}/2|\mathcal{F}_{T-1})>0$$ almost surely but this contradicts that the conditional law of 
$X_{T}-h_{*}$ with respect to $\mathcal{F}_{T-1}$ is degenerate.
\end{proof}

\begin{example}\label{poption}{\rm  
Consider now the \emph{put options} 
$\{(K-X_{T})^{+}: K_{\flat}\leq K\leq K_{\sharp}\}$.
Take a mapping $Q:[K_{\flat},K_{\sharp}]\to \mathcal{M}$ such that 
$q:=\inf_{K}E_{Q(K)}[(K-X_{T})^{+}]>0$ (that is, the prices of the put options cannot be arbitrarily small,
a rather natural assumption).

Define $$H_{2}:=\{(K-X_{T})^{+}-E_{Q(K)}[(K-X_{T})^{+}]: K_{\flat}\leq K\leq K_{\sharp}\},$$
set $\tilde{H}_{2}:=\overline{\mathrm{conv}(H)}$ and $L_{2}:=\mathcal{C}(\tilde{H}_{2})$.}
\end{example}

\begin{proposition}\label{putt} If $X_{T}$ is unbounded from above then $L_{2}$ is a small cone.
\end{proposition}
\begin{proof}
As $-K_{\sharp}\leq h\leq K_{\sharp}$ for all 
$h\in H_{2}$, $\tilde{H}_{2}$ is bounded. 
Notice that $\delta:=P(X_{T}\geq K_{\sharp})$ is positive.
On the event $\{X_{T}\geq K_{\sharp}\}$, $$
(K-X_{T})^{+}-E_{Q(K)}[(K-X_{T})^{+}]\leq -E_{Q(K)}[(K-X_{T})^{+}]
\leq -q,$$ for all $K$.
Hence for any $h\in {\mathrm{conv}(H_{2})}$, $P(h\leq -q)\geq \delta$.{}
If $h_{n}\in\mathrm{conv}(H_{2})\to h_{*}$ in probability as $n\to\infty$, necessarily $h_{*}\neq 0$.
It follows that $0\notin\tilde{H}_{2}$ and  $\mathcal{C}(\tilde{H}_{2})$ is a small cone.
\end{proof}

\begin{proposition}\label{putin} 
Assume that the $\mathcal{F}_{T-1}$-conditional support of $X_{T}$ almost surely contains $[0,\infty)$.
Then $K_{0}\cap L_{2}=\{0\}$.
\end{proposition}
\begin{proof}
As in the proof of Proposition \ref{callin}, suppose that 
$$
h_{n}:=\sum_{j=1}^{k(n)}\lambda^{n}_{j}[(K^{n}_{j}-X_{T})^{+}-E_{Q(K_{j}^{n})}[(K^{n}_{j}-X_{T})^{+}]]\to h_{*}.$$
Clearly, $h_{*}$ is a bounded random variable hence, as before, $h_{*}$ must be $\mathcal{F}_{T-1}$-measurable.
We will show that this cannot be the case.

Arguing on the event $\{X_{T}\geq K_{\sharp}\}$, 
$$
-\sum_{j=1}^{k(n)}\lambda^{n}_{j}E_{Q(K_{j}^{n})}[(K^{n}_{j}-X_{T})^{+}]\to d_{*},\ n\to\infty 
$$
for some constant $d_{*}$. Hence $$
P(h_{*}=d_{*}|\mathcal{F}_{T-1})\geq P(X_{T}\geq K_{\sharp}|\mathcal{F}_{T-1})>0\mbox{ a.s.}
$$
At the same time, on the event $\{0\leq X_{T}\leq K_{\flat}/2\}$, $h_{*}\geq K_{\flat}/2+d_{*}$
and $$P(0\leq X_{T}\leq K_{\flat}/2|\mathcal{F}_{T-1})>0\mbox{ a.s.}$$
This contradicts that $h_{*}$ has a degenerate conditional law with respect to $\mathcal{F}_{T-1}$.
\end{proof}

We have thus shown in Propositions \ref{calloption}, \ref{callin}, \ref{putt} and
\ref{putin} that choosing $\mathcal{Y}$ as the convex closure of either call or put options, 
the conditions of Theorems \ref{nao} and \ref{util} are met. A remaining issue is the verification
of the $\mathbf{NA}(\mathcal{Y})$ condition.

\begin{remark}\label{noo}{\rm We would like to point out here that, in the most general case, this is an extremely difficult question.
Note that in Examples \ref{coption} and \ref{poption} the pricing measures $Q(K)$ for various strike
prices $K$ can be \emph{arbitrary}, even the measurability of $Q(\cdot)$ is not required. 

Let us stay with put options for the moment. Given
a (possibly weird) set of pricing measures $Q(K)$ it seems quite hopeless to decide whether there is
a $Q_{*}\in\mathcal{M}_{\infty}(\mathcal{Y})$ which satisfies $E_{Q_{*}}[(K-X_{T})^{+}]\leq E_{Q(K)}[(K-X_{T})^{+}]$
for all $K$ (and we know from Theorem \ref{nao} that the existence of such a measure is
a necessary condition for $\mathbf{NA}(\mathcal{Y})$). 

One could clearly ``cheat'' by taking $Q(K)=Q$ for all $K$ and then, under mild assumptions,
one can easily verify that $Q\in\mathcal{M}(\mathcal{Y})$ and hence $\mathbf{NA}(\mathcal{Y})$ is true.
However, in an over-the-counter market it seems highly implausible to us that the same $Q$
would be used to price every call option, different sellers may well have
different pricing measures, this does not necessarily contradict absence of arbitrage in the market.}
\end{remark}

The difficulties explained in Remark \ref{noo} notwithstanding,
we provide a simple but non-trivial condition ensuring $\mathbf{NA}(\mathcal{Y})$.

\begin{proposition}\label{compacta} In the setting of Example \ref{poption}, let
the $\mathcal{F}_{t-1}$-conditional support of $X_{t}-X_{t-1}$ be $\mathbb{R}$, for all $t=1,\ldots,T$.{}
Then $\mathbf{NA}(\mathcal{Y})$ holds.
\end{proposition}
Clearly, the conditions of Proposition \ref{compacta} imply those of Propositions \ref{putt} and \ref{putin}
so $L_{2}$ is a small cone with $K_{0}\cap L_{2}=\{ 0\}$.
\begin{proof}[Proof of Proposition \ref{compacta}] 
Let $(\phi,h)\in\mathcal{A}$ be such that $$V(\phi,h)=\sum_{j=1}^{T}\phi_{j}(X_{j}-X_{j-1})+h\geq 0$$ almost surely.
Note that, by our assumption on the conditional supports, 
$\sum_{j=1}^{T}\phi_{j}(X_{j}-X_{j-1})$ is unbounded from below unless $\phi_{j}=0$ almost surely for $j=1,\ldots,T$.{}

Since $h$ is bounded, $V(\phi,h)\geq 0$ implies that, indeed, one must necessarily have $\phi_{j}=0$, $j=1,\ldots,T$.{}
If $h\neq 0$ then $h=\alpha h_{*}$ where $\alpha>0$ and $h_{*}\in \tilde{H}_{2}$. It follows that $h_{*}\geq 0$
almost surely which is impossible by the proof of Proposition \ref{putt}. So the only remaining possibility is $h=0$:
$\mathbf{NA}(\mathcal{Y})$ holds.
\end{proof}

\begin{remark}\label{conil}{\rm 
Our research suggests that, in markets with both dynamical assets and static options,
a satisfactory arbitrage theory can be built up provided that the 
set of options is not too large, as described by the concept
of small cones. 

A weak point of our results is that we cannot cover cases where the set of available options is
a linear space, see Remark \ref{notall} above. Generalizations in this direction are the first
objectives of our future work.

Extending Theorem \ref{util} to unbounded utilities requires more involved arguments and will be done elsewhere.
The use of small cones in continuous-time models is an exciting direction of future research that could be pursued.}
\end{remark}

\noindent\textbf{Declarations.} 
The author gratefully 
acknowledges the support of the Hungarian 
National Research, Development and Innovation Office (NKFIH) through grant K 143529.

\smallskip{}

\noindent The author has no relevant financial or non-financial interests to disclose.

\smallskip{}

\noindent There is no data associated to this research.

\end{document}